\documentclass[journal,draftcls,onecolumn,12pt,twoside]{IEEEtranTCOM}
\normalsize
\ifCLASSINFOpdf
\else
\fi
\usepackage{amsmath}
\usepackage{graphicx}
\usepackage{verbatim}
\usepackage{amssymb,amsmath}
\usepackage{mathtools}
\usepackage{color}
\usepackage{epstopdf}
\usepackage{bm}
\usepackage[nosort]{cite}
\usepackage{array,color}

\usepackage{algorithm}
\usepackage{algpseudocode}
\usepackage{amsmath}
\DeclareMathOperator*{\argmax}{arg\,max}
\newtheorem{lemma}{\textbf{Lemma}}
\newtheorem{corollary}{\textbf{Corollary}}
\newtheorem{Theorem}{\textbf{Theorem}}
\newtheorem{Remark}{\textbf{Remark}}
\usepackage{graphics}
\usepackage{epsfig}
\usepackage[normalem]{ulem}
\usepackage{cite}
\usepackage{multicol}
\usepackage{balance}
\usepackage{caption}
\begin{document}
%
\title{Design and Analysis of Online Fountain Codes for Intermediate Performance}
%
%
%
\author{Jingxuan~Huang,
        Zesong~Fei, 
        Congzhe~Cao,
        and Ming~Xiao 
\thanks{This work was supported by the National Natural Science Foundation of China under Grant No.61871032, by Chinese Ministry of Education-China Mobile Communication Corporation Research Fund under Grant MCM20170101, and by the 111 Project of China under Grant B14010. (\emph{Corresponding author: Zesong Fei})}
\thanks{J. Huang and Z. Fei are with the School of Information and Electronics, Beijing Institute of Technology, Beijing 100081, China (e-mail: 1120121556@bit.edu.cn, feizesong@bit.edu.cn).}%
\thanks{C. Cao is with the Department of Electrical and Computer Engineering, University of Alberta, Edmonton, AB T6G 1H9, Canada (email: congzhe@ualberta.ca).}%
\thanks{M. Xiao is with School of Electrical Engineering and Computer Science, KTH Royal Institute of Technology, Stockholm 100 44, Sweden (e-mail: mingx@kth.se).}
}

%
%

\markboth{IEEE Transactions on Communications}%
{Submitted paper}

%


\maketitle

\begin{abstract}
For the benefit of improved intermediate performance, recently online fountain codes attract much research attention. However, there is a trade-off between the intermediate performance and the full recovery overhead for online fountain codes, which prevents them to be improved simultaneously. We analyze this trade-off, and propose to improve both of these two performance. We first propose a method called Online Fountain Codes without Build-up phase (OFCNB) where the degree-1 coded symbols are transmitted at first and the build-up phase is removed to improve the intermediate performance. Then we analyze the performance of OFCNB theoretically. Motivated by the analysis results, we propose Systematic Online Fountain Codes (SOFC) to further reduce the full recovery overhead. Theoretical analysis shows that SOFC has better intermediate performance, and it also requires lower full recovery overhead when the channel erasure rate is lower than a constant. Simulation results verify the analyses and demonstrate the superior performance of OFCNB and SOFC in comparison to other online fountain codes.
\end{abstract}
\nopagebreak[4]

%
\IEEEpeerreviewmaketitle

\section{Introduction}
%
%
%
%
\IEEEPARstart{F}{ountain} codes are a class of error control codes that can generate an infinite number of coded symbols from a limited number of source symbols. The code rates of fountain codes are not fixed. Thus these codes can adapt to various channel qualities. The concept of fountain codes was first introduced in \cite{Byers1998DF} to distribute packet data in erasure channels. Luby Transform (LT) codes \cite{Luby2002LT} are the first practical fountain coding schemes that are capacity-achieving for binary erasure channels (BECs). Then Raptor codes \cite{shokrollahi2006raptor} were introduced to improve the error-floor performance and reduce decoding complexity. Because of the low complexity and excellent performance, fountain codes are used in multimedia transmission \cite{cao2013extended} and relay networks \cite{Hussain2014BDLT}.

The performance of fountain codes is determined by the selection rule of source symbols, which is usually described by degree distributions of source symbols and coded symbols. Luby introduced the robust solition distribution (RSD) as the coded symbol degree distribution in \cite{Luby2002LT}. The scheme in \cite{Sorensen2012ripplesize} proposed a new coded symbol degree distribution, which provides a smaller ripple size to reduce decoding overhead. The scheme in \cite{Hussain2013MBLT} proposed a modified the degree distribution of source symbols that maximized the minimum degree to improve the error-floor performance. The schemes in \cite{Hayajneh2015MBLT} and \cite{Me2017SMBLT} modified the degree distributions of source symbols to improve the decoding convergence speed. Systematic fountain codes \cite{Yuan2008SLT,Xu2016SLT,Okpotse2019SFC} can also be considered as modification on degree distributions of coded symbols. In \cite{Yuan2008SLT}, the authors proposed systematic LT codes to reduce the encoding and decoding complexity. The scheme in \cite{Xu2016SLT} applied systematic LT codes in binary input additive white Gaussian noise (BIAWGN) channels and optimized the degree distribution of coded symbols. In \cite{Okpotse2019SFC}, systematic fountain codes with a truncated Poisson degree distribution were applied in distributed storage systems.

Recently the online property of fountain codes attracts much attention. The online property represents that once given an instantaneous decoding state, the encoder can find the optimal coding strategy efficiently. The online property can thus avoid the sub-optimal performance caused by the pre-defined and fixed encoding process of conventional fountain codes. Some codes with online property were proposed, such as the Growth codes \cite{kamra2006growth} and real-time oblivious codes \cite{beimel2007rt}. However, the decoding overhead of these codes is relatively large. In \cite{cassuto2015online}, Cassuto and Shokrollahi proposed a new class of online fountain codes, which utilize feedback to reduce decoding overhead significantly. The decoding state is presented by the uni-partite graph in \cite{cassuto2015online}, and the encoding scheme is divided into two phases called the build-up phase and completion phase. Based on \cite{cassuto2015online}, the schemes in \cite{Me2017UEPonline} and \cite{Cai2019URT} provided unequal error protection for online fountain codes. In \cite{huang2017improved}, the minimum degree of source symbols was maximized for online fountain codes to achieve lower full recovery overhead and less feedback. In \cite{Zhao2018IETonline}, coded symbols that contain three unrecovered source symbols were utilized for online fountain codes to reduce overhead. A new theoretical analysis framework was proposed in \cite{Me2018IOFC} to more accurately describe the relationship between the number of coded symbols transmitted and the number of source symbols recovered. In \cite{Yi2018onlineData}, online fountain codes were used in wireless sensor networks for data gathering.

One interesting problem of online fountain codes is how to reduce the full-recovery overhead and improve the intermediate performance at the same time. On the one hand, the schemes in \cite{huang2017improved} and \cite{Me2018IOFC} have lower overhead but worse intermediate performance than the original online fountain codes. However, the scheme in \cite{hashemi2016fountain} has better intermediate performance but requires more full recovery overhead. The authors of \cite{Cai2019URT} and \cite{Zhao2018IETonline} attempted to transmit degree-1 coded symbols first to improve intermediate performance, but the effect on full recovery overhead was not analyzed.

In this paper, we further investigate the trade-off between full recovery overhead and intermediate performance. In the proposed scheme, i.e., Online Fountain Codes without Build-up phase (OFCNB), the degree-1 coded symbols are transmitted first and the build-up phase in the original online fountain codes is removed. By exploiting the theoretical framework in \cite{Me2018IOFC}, the intermediate performance and the full recovery overhead of the proposed OFCNB are analyzed. Inspired by the analysis results, we further propose Systematic Online Fountain Codes (SOFC) to improve the intermediate performance and reduce the full recovery overhead. The performance of SOFC is also analyzed and compared with the scheme in \cite{cassuto2015online}. Simulation results validate the accuracy of our analyses and demonstrate the performance improvement of OFCNB and SOFC compared with different online fountain codes.

The rest of this paper is organized as follows. Section II briefly reviews the encoding and decoding process of on-line fountain codes in \cite{cassuto2015online}. We propose the OFCNB and analyze its full recovery overhead and intermediate performance in Section III. In Section IV, we improve OFCNB by making it systematic and analyze its performance. Section V presents simulation results. Finally Section VI concludes the paper.

\section{Online Fountain Codes}
In this section, we briefly review the uni-partite graph and the online fountain coding scheme proposed in \cite{cassuto2015online}.
For the online fountain codes, only the received coded symbols belonging to the following cases are used for decoding, and others are discarded.
\begin{itemize}
\item \textit{Case-1}:
 The degree-$m$ coded symbol that is the XOR of one unrecovered source symbol and $m-1$ recovered source symbols.
\item \textit{Case-2}:
The degree-$m$ coded symbol that is the XOR of two unrecovered source symbol and $m-2$ recovered source symbols.
\end{itemize}
The uni-partite graph is used to demonstrate the decoding state. Thus the decoding process can be monitored and adjusted. Different from the conventional bi-partite graph, which uses check nodes for coded symbols and variable nodes for source symbols, the nodes in uni-partite graph are only used to present source symbols, which are denoted as source nodes, as shown in Fig. \ref{unigraph}. If a source symbol is recovered at the receiver, the corresponding source node is colored black, otherwise it is colored white. If a coded symbol is the XOR of two source symbols, the two source nodes are connected by an edge. A component in a uni-partite graph is a set of white source nodes that are connected by edges, and none of which is connected to any node outside the component. The number of nodes that a component contains is called the component size. A white source node not connected by any edge is a component of size-1. Once the receiver recovers a source symbol, the corresponding white source node is colored black, then the edges connecting to this node are removed, and its neighboring nodes are colored black, too. Thus all the source nodes in the component turn black once one of the corresponding source symbol is recovered.

\begin{figure}[t]
\setlength{\abovecaptionskip}{0.cm}
\setlength{\belowcaptionskip}{-0.8cm}
\centering
\includegraphics[scale=0.35]{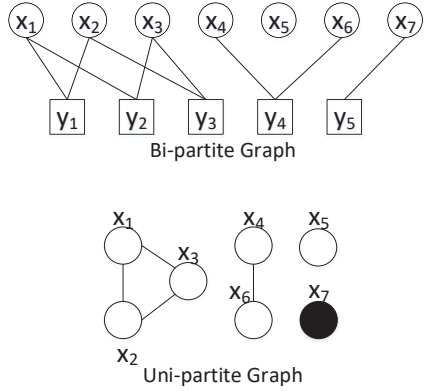}
\caption{Bi-partite graph and the corresponding uni-partite graph.}
\label{unigraph}
\end{figure}

The encoding process of the online fountain codes consists of the build-up phase and the completion phase, which is briefly described as follows.

\subsection {Build-Up phase}
We denote $k$ the number of source symbols, and the fractional size of the largest component is denoted as $\beta_0$ ($0 < \beta_0 < 1$). The initial decoding graph is an uni-partite graph with $k$ white nodes and no edges. First, coded symbols with degree-2 are transmitted to add edges and connect nodes in the decoding graph. Once the size of the largest component becomes $k\beta_0$, the receiver sends a feedback to the transmitter, and the transmitter starts to send degree-1 coded symbols, until the largest component turns black.
\subsection {Completion phase}
At the completion phase, the decoder at the receiver side feeds back the decoding state to the transmitter. The encoder at the transmitter side generates coded symbols based on the instantaneous decoding state and transmits them to the receiver. The decoder then processes those coded symbols, updates the decoding graph and feeds back the updated decoding state to the transmitter. The completion phase continues until all source nodes in the decoding graph turn black. That is, all source symbols are recovered at the receiver side. The details of the completion phase are as follows.
\begin{itemize}
\item \textbf{Transmitter side}:
Given $k\beta$ black nodes where $\beta$ is the percentage of recovered source symbols, the encoder calculates the optimal degree $\widehat{m}$ that maximize the probability that the coded symbol is a Case-1 coded symbol or a Case-2 coded symbol. Then the encoder selects $\widehat{m}$ source symbols uniformly and randomly to form a coded symbol and sends it to the receiver. The optimal degree $\widehat{m}$ satisfies
\begin{equation}
\widehat{m} = \argmax_m [P_{1}(m,\beta)+P_{2}(m,\beta)],
\label{degree}
\end{equation}
where $P_{1}(m,\beta)$ and $P_{2}(m,\beta)$ are the probabilities that a degree-$m$ coded symbol belongs to Case-1 or Case-2, respectively. $P_{1}(m,\beta)$ and $P_{2}(m,\beta)$ can be evaluated as
\begin{gather}
P_{1}(m,\beta)=\binom m1 \beta^{m-1}(1-\beta),\\
P_{2}(m,\beta)=\binom m2 \beta^{m-2}(1-\beta)^2.
\end{gather}

\item \textbf{Receiver side}: A coded symbol is utilized to update the decoding graph if it belongs to Case-1 or Case-2, otherwise it is discarded. The decoder calculates the recovery rate $\beta$ once the decoding graph is updated, and feeds it back to the encoder when the optimal degree changes because of an updated $\beta$.
\end{itemize}

\section{Online fountain codes without build-up phase}
In this section, we first propose an online fountain coding scheme without build-up phase (OFCNB), which has better intermediate performance than the conventional online fountain codes (OFC) in \cite{cassuto2015online}. Moreover, we analyze the performance of OFCNB in two special cases based on the framework in \cite{Me2018IOFC}, and then extend the performance analysis to general cases.
\subsection{The Proposed coding scheme}
From Section II, it is clear that for OFC, no source node can turn black before degree-1 coded symbols are transmitted at the end of the build-up phase. That is, no source symbol can be recovered before the end of the build-up phase, which degrades the intermediate performance. If the degree-1 coded symbols are transmitted first, the processing of building a large component will result in a partial recovery of the source symbols, and the intermediate performance can be improved. In addition, Fig. \ref{pmax} demonstrates the probability that a degree-$m$ coded symbol belongs to Case-1 and Case-2 with varying recovery ratio $\beta$. From Fig. \ref{pmax} we can see that when recovery rate $\beta \leq 0.5$, the optimal degree is 2. Thus if we send degree-1 symbols beforehand and remove the build-up phase, the encoder in the completion phase behaves still the same as the build-up phase when $\beta<0.5$. Moreover, note that the decoder holds the same rule for both build-up phase and completion phase. Thus the evolution of the decoding graph will stay the same. Motivated by the aforementioned fact, we introduce OFCNB as follows.

First, the encoder generates and transmits uniformly and randomly selected degree-1 coded symbols, until the decoder at the receiver recovers $k\gamma_0$ source symbols, where $\gamma_0$ ($0<\gamma_0\leq 1$) is a parameter that can be adjusted. Then the receiver sends a feedback to the transmitter side, and the encoder generates coded symbols with optimal degree $\widehat{m}$ by substituting $\beta=\gamma_0$ in (\ref{degree}). The rest of the encoding process is the same as the completion phase described in Subsection B of Section II.

By adjusting the value of $\gamma_0$, we can achieve trade-off between the intermediate performance and the full recovery overhead. In general, a larger $\gamma_0$ means better intermediate performance but larger recovery overhead, and vice versa. The impact of parameter $\gamma_0$ will be discussed in detail in the following subsections.

\begin{figure}[t]
\centering
\includegraphics[scale=0.4]{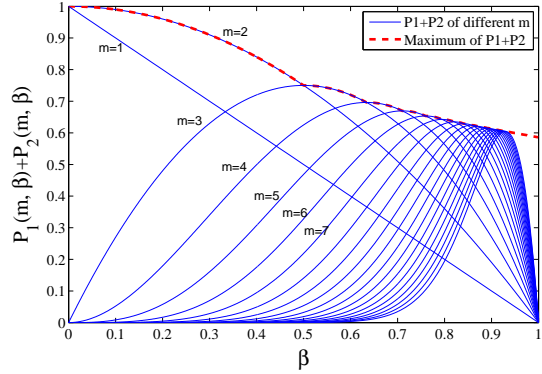}
\caption{$P_1(m,\beta)+P_2(m,\beta)$ and $max(P_1(m,\beta)+P_2(m,\beta))$ (part of the curves are from \cite{Me2018IOFC}).}
\label{pmax}
\end{figure}

\subsection{Performance Analysis for Small $\gamma_0$}
We now provide performance analysis for the proposed OFCNB based on the framework in \cite{Me2018IOFC}, which we will briefly review for the purpose of a complete discussion as follows. In this subsection we will focus on the case that the parameter $\gamma_0$ is very small, i.e., $\gamma_0\to 0$.

At the beginning of OFCNB, randomly selected degree-1 coded symbols are transmitted. Denote the expected number of degree-1 coded symbols that are required to recover $s_1$ source symbols as $E(N_{small1}(s_1))$, we can have the following theorem.
\begin{Theorem}
$E(N_{small1}(s_1))$ is given by
\begin{equation}
E(N_{small1}(s_1))=s_1.
\end{equation}
\end{Theorem}
\begin{proof}
When $\gamma_0\to 0$, the probability that a source symbol is selected more than once is negligible. Thus one degree-1 coded symbol recovers one source symbol once it is received at the decoder.
\end{proof}

Then the encoding process comes to the completion phase. We divide the newly proposed completion phase into two stages, i.e., $\gamma_0<\beta \leq 0.5$ for the first stage and $0.5<\beta \leq 1$ for the second stage. For the first stage, the relationship between the number of recovered source symbols and the expected number of coded symbols required can be described as follows.

\begin{Theorem}
Denote the number of required source symbols for recovering $s_2$ coded symbols at the first stage as $N_{small2}(s_2)$, the expectation of $N_{small2}(s_2)$ is expressed as follows:
\begin{equation}
E(N_{small2}(s_2))=-\frac{k^2\ln(1-\frac{s_2}{k})}{2s_2}.
\end{equation}
\end{Theorem}
\begin{proof}
For the first stage, from Fig. \ref{pmax} we can see that only degree-2 coded symbols are generated and transmitted. Science $\gamma_0$ is small and can be neglected, the evolution of the decoding graph for the first stage is similar to the build-up phase of OFC.

For the build-up phase of OFC, the evolution of the decoding graph is affected by the average source symbol degree $c$, which is the average number of times that a source symbol is selected. When $0<c\leq 1$, several components of size $\Theta(\ln(k))$ will be built up with degree-2 coded symbols received. When $c=1$, the small components start to join rapidly to form one giant component, which is called phase transition. When $c>1$, this giant component becomes the largest component of size $\Theta(k)$, where other components are still small with size $\Theta(\ln(k))$ \cite{alon2004probabilistic}. The fractional size of the giant component $\alpha$ and $c$ satisfies the following constraint
\[
\alpha+e^{-c\alpha}=1,
\]
where $0<\alpha<1$ and $c>1$, and thus
\[
c=-\frac{\ln(1-\alpha)}{\alpha}.
\]
Here $N$, the number of required degree-2 coded symbols for building a giant component of fractional size $\alpha$, can be expressed as
\begin{equation}
N=\frac12kc=-\frac{k\ln(1-\alpha)}{2\alpha},
\label{N}
\end{equation}
since $kc$ is the total number of times that source symbols are selected and the coefficient $\frac 12$ means that a degree-2 coded symbol corresponds to two source symbols.

For the first stage of the proposed completion phase, because $\gamma_0k$ source symbols are already recovered, and the probability that they will connect to the giant component is $1-(1-\alpha)^{\gamma_0k}$, which is high in practice. Thus the giant component is very likely to turn black once it is built. When other small component joins the giant component, they are colored black, too. Thus the size of the giant component can be considered as the number of recovered source symbols, and we have
\[
\alpha=\frac{s_2}{k}.
\]
Substitute $\alpha$ into (\ref{N}), we have
\begin{align*}
E(N_{small2}(s_2))&=N
=-\frac{k\ln(1-\frac{s_2}{k})}{2\frac{s_2}{k}}
=-\frac{k^2\ln(1-\frac{s_2}{k})}{2s_2}.
\end{align*}
\end{proof}

For the second stage of the proposed completion phase, the relationship between the number of recovered source symbols and the number of required coded symbols is the same as the conventional completion phase in OFC. First we introduce the following two lemmas from \cite{Me2018IOFC}.
\begin{lemma}
At the completion phase, on average the recovery of one source symbol needs one coded symbol that is \emph{useful}. We list the three types of \emph{useful} symbols as follows.
\begin{itemize}
\item \text{Build-Up edge}: A coded symbol that corresponds to an edge in small components that are not connected to the giant component.
\item \text{Case-1 completion symbol}: A coded symbol that belongs to Case-1 at the completion phase.
\item \text{Case-2 completion symbol}: A coded symbol that belongs to Case-2 at the completion phase.
\end{itemize}
\end{lemma}

It is clear that Build-Up edges have connected two components into a single component, Case-2 completion symbols can connect two components, and a Case-1 completion symbol can turn a component into black.
\begin{lemma}
Denote $N_{c1,2}(n)$ the number of Case-1 and Case-2 completion symbols required for recovering $n$ source symbols at the completion phase with $\beta_0=0.5$. The expectation of $N_{ca1,2}(n)$ is expressed as follows:
\begin{equation}
E(N_{c1,2}(n))=n-E(N_{BU}(n))=(1-\frac{1}{4}c_0)n,
\end{equation}
where $E(N_{BU}(n))=\frac{1}{4}c_0n$ is the expected number of Build-Up edges connected to $n$ source symbols and $c_0=-\frac{\ln(1-\beta_0)}{\beta_0}=1.3863$.
\end{lemma}

Now we propose Theorem 3 to analyze the second stage of the proposed completion phase as follows.

\begin{Theorem}
Denote the number of coded symbols required for recovering $s_3$ source symbols at the second stage as $N_{small3}(s_3)$, the expectation of $N_{small3}(s_3)$ is expressed as
\begin{equation}
E(N_{small3}(s_3))=(1-\frac 14c_0)\cdot \sum_{i=\frac12k}^{s_3+\frac12k-1}{\frac{1}{P_M(i)}}.
\end{equation}
The function $P_M(n)$, which denotes the probability that a transmitted coded symbol is a Case-1 or Case-2 completion symbol when $n$ source symbols are recovered, can be expressed as
\begin{equation}
P_M(n)=P_{1}(\widehat{m},\frac{n}{k})+P_{2}(\widehat{m},\frac{n}{k}),
\end{equation}
where $\widehat{m}$ is defined as (\ref{degree}), $P_1(m,\beta)$ and $P_2(m,\beta)$ are defined as (2) and (3), respectively.
\end{Theorem}
\begin{proof}
Note that in the first stage  build-up edges will also be introduced. From Lemma 1 and Lemma 2 we can find that in order to recover one source symbol, there are
\[
E(N_{ca1,2}(n))-E(N_{ca1,2}(n-1))=1-\frac{1}{4}c_0
\]
Case-1 or Case-2 completion symbols required. When $s_0$ source symbols are recovered at the second stage, in total $s_0+\frac12k$ source symbols are recovered, and the probability that a coded symbol is a Case-1 or Case-2 completion symbol can be expressed as
\[
P_M(s_0+\frac12k)=P_{1}(\widehat{m},\frac{s_0+\frac12k}{k})+P_{2}(\widehat{m},\frac{s_0+\frac12k}{k}).
\]
Thus $(1-\frac{1}{4}c_0)/P_M(s_0+\frac12k)$ coded symbols are required for recovering one source symbol. Now we have $s_0+1$ source symbols recovered in the second stage. Continue this recovery process and by using induction we can obtain (8).
The detailed proof can be found in Section III from \cite{Me2018IOFC}.
\end{proof}

Finally we can get the relationship between the number of recovered source symbols and the number of required coded symbols for the entire coding process when $\gamma_0\to 0$. Combing the three theorems above we present the following remark.
\begin{Remark}
When $\gamma_0\to 0$, denote the number of required coded symbols for recovering $s$ source symbols as $N_{s}(s)$, the expectation of $N_{s}(s)$ is expressed as follow:
\begin{align}
E(N_{s}(s))= &\left\{
\begin{aligned}
&s,\text{when } 0<s \leq \gamma_0k;\\
&-\frac{k^2\ln(1-\frac{s-\gamma_0k}{k})}{2(s-\gamma_0k)},\text{when } \gamma_0k<s\leq \frac12k;\\
&(1-\frac 14c_0)\cdot \sum_{i=\frac12k}^{s-1}{\frac{1}{P_M(i)}}-\frac{k\ln(\frac12+\gamma_0)}{1-2\gamma_0},\text{when } \frac12k<s \leq k.
\end{aligned}
\right.
\end{align}
\end{Remark}

Note that before stage 1 of the proposed completion phase, $\gamma_0k$ source symbols have been recovered. Before stage 2 of the proposed completion phase, $\frac12 k$ source symbols have been recovered.

For OFC in \cite{cassuto2015online} with $\beta_0=0.5$, we have
\begin{equation}
E(N_{OFC}(s))=k\ln(2)+(1-\frac 14c_0)\cdot \sum_{i=\frac12k}^{s-1}{\frac{1}{P_M(i)}},
\end{equation}
when $s>0.5$ \cite{Me2018IOFC}, where $N_{OFC}(s)$ is the number of required coded symbols for recovering $s$ source symbols in OFC. Note that no source symbols are recovered before $k\ln(2)$ coded symbols are transmitted. Comparing (10) and (11), we have following corollary.
\begin{corollary}
OFCNB has better intermediate performance and the same full recovery overhead compared with OFC when $\gamma_0\to 0$.
\end{corollary}
\subsection{Performance Analysis for Large $\gamma_0$}
After we obtain performance analysis for the case where $\gamma_0\to 0$, we now analyze another case where $0.5 \leq \gamma_0 \leq 1 $.

A large $\gamma_0$ means many degree-1 symbols are transmitted at first, which results in good intermediate performance. However, with an increasing number of source symbols recovered, the probability that a randomly selected degree-1 coded symbol is a recovered source symbol also becomes larger. Note that duplicate coded symbols are not helpful for decoding, which results in larger full recovery overheads. First we present the following theorem that describes the number of randomly selected degree-1 coded symbols required for recovering $s_1$ ($s_1<\gamma_0k$) source symbols.
\begin{Theorem}
Denote the number of required degree-1 coded symbols for recovering $s_1$ ($0<s_1 \leq \gamma_0k$) source symbols as $N_{large1}(s_1)$, the expectation of $N_{large1}(s_1)$ can be expressed as follows:
\begin{equation}
E(N_{large1}(s_1))=k\ln(\frac{k}{k-s_1}).
\end{equation}
\end{Theorem}
\begin{proof}
Because of random selection, the degree distribution of source symbols follows Poisson distribution \cite{alon2004probabilistic}. The probability that a source symbol has a degree $d$ can be expressed as
\[
P(X=d)=\frac{c^d}{d!}e^{-c},
\]
where $c$ is the average degree of source symbols. Note that only degree-1 coded symbols are generated. That is, once a source symbol is selected to form a coded symbol, it can be recovered right away at the receiver side. Thus only degree-0 source symbols cannot be recovered, and we have
\[
E(\frac{s_1}{k})=P(X\neq0)=1-P(X=0)=1-\frac{c^0}{0!}e^{-c}=1-e^{-c}.
\]
Because a degree-1 coded symbol will increase the average source symbol degree by 1, we have
\[
c=\frac{N_{large1}}{k},
\]
thus
\[
E(\frac{s_1}{k})=1-e^{-\frac{N_{large1}}{k}},
\]
and Theorem 4 is obtained.
\end{proof}

We can also find that because the coded symbol degree is 1, all source symbols with degree $d>1$ generates at least one duplicate coded symbol.

\begin{Theorem}
At the completion phase, denote the number of coded symbols required for recovering $s_2$ ($0< s_2\leq (1-\gamma_0)k$) source symbols as $N_{large2}(s_2)$, the expectation of $N_{large2}(s_2)$ can be expressed as
\begin{equation}
E(N_{large2}(s_2))=\sum_{i=\gamma_0k}^{s_2+\gamma_0k-1}{\frac{1}{P_M(i)}}.
\end{equation}
\end{Theorem}
\begin{proof}
Based on Lemma 1, there are three types of \emph{useful} coded symbols, and on average the recovery of one source symbol needs one \emph{useful} coded symbol. For the proposed completion phase, when $\gamma_0\geq0.5$, no degree-2 coded symbol is generated. Thus the \emph{useful} coded symbols only consist of Case-1 completion symbols and Case-2 completion symbols. When $s_0$ source symbols are recovered in the proposed completion phase, $s_0+\gamma_0k$ source symbols are recovered in total, and the probability that a coded symbol is a Case-1 completion symbol or a Case-2 completion symbol is
\[
P_M(s_0+\gamma_0k)=P_{1}(\widehat{m},\frac{s_0+\gamma_0k}{k})+P_{2}(\widehat{m},\frac{s_0+\gamma_0k}{k}).
\]
In order to recover one more source symbol, the expectation of the required coded symbol is $1/P_M(s_0+\gamma_0k)$. Continue this recovery process until we have $s_2+\gamma_0k$ source symbols recovered, we can obtain Theorem 5.
\end{proof}

Combing Theorem 4 and Theorem 5, we can obtain the analysis for the entire encoding process.
\begin{Remark}
When $0.5\leq \gamma_0 \leq 1$, denote the number of coded symbols required for recovering $s$ source symbols as $N_{l}(s)$, the expectation of $N_{l}(s)$ is expressed as
\begin{align}
E(N_{l}(s))=
&\left\{
\begin{aligned}
&k\ln(\frac{k}{k-s}),&\text{when } 0<s\leq \gamma_0k;\\
&\sum_{i=\gamma_0k}^{s-1}{\frac{1}{P_M(i)}}-k\ln(1-\gamma_0),&\text{when } \gamma_0k<s \leq k.
\end{aligned}
\right.
\end{align}
\end{Remark}

Comparing the results in Remark 1 and Remark 2, we can obtain the following theorem.
\begin{Theorem}
Denote the number of coded symbols required to recover $\frac12k$ source symbols as $N_{\frac12k}$. For an arbitrary $\gamma_0$, we have
\begin{equation}
E(N_{\frac12k})=k\ln(2).
\end{equation}
\end{Theorem}
\begin{proof}
For the case that $\gamma_0$ is small enough to be neglected, only degree-2 coded symbols are generated before $s=\frac12 k$. For the case that $\gamma_0$ is larger than 0.5, only degree-1 coded symbols are generated before $s=\frac12 k$. The rest of situations are intermediate states of these two special cases, where both degree-2 coded symbols and degree-1 coded symbols are generated. Thus the number of required coded symbols to recover $\frac12k$ source symbols is bounded by $E(N_s(\frac12k))$ and $E(N_{l}(\frac12k))$, i.e.,
\[
E(N_s(\frac12k))\leq E(N_{\frac12k})\leq E(N_{l}(\frac12k)).
\]
When $\gamma_0\to0$, from Remark 1 we have
\[
E(N_s(\frac12k))=-\frac{k\ln(\frac12+\gamma_0)}{2(\frac12-\gamma_0)},
\]
by neglecting $\gamma_0k$ we have
\[
E(N_s(\frac12k))=-\frac{k\ln(\frac12)}{2\cdot\frac12}=-k\ln(\frac12)=k\ln(2).
\]
On the other hand, when $0.5\leq \gamma_0 \leq 1$, from Remark 2 we have
\[
E(N_{l}(\frac12k))=k\ln(\frac{k}{k-\frac12k})=k\ln(2)=E(N_s(\frac12k)).
\]
Thus
\[
E(N_s(\frac12k))=E(N_{\frac12k})=E(N_{l}(\frac12k))=k\ln(2).
\]
\end{proof}

Based on Theorem 6, we can discuss the effect of $\gamma_0$ in detail. First we focus on the process when $\beta<0.5$. If $\gamma_0$ is small, the encoding process will suffer from the phase transition, which degrades the intermediate performance. If $\gamma_0<0.5$, a larger $\gamma_0$ means more degree-1 coded symbols are transmitted that can be decoded immediately, and thus better intermediate performance. If $\gamma_0>0.5$, the extra degree-1 coded symbols transmitted will not positively affect the intermediate performance when $\beta<0.5$. Thus the intermediate performance is the same as $\gamma_0=0.5$.

When $\beta=0.5$, from Theorem 6, we can find the number of required coded symbols for recovering half of the source symbols is the same for different $\gamma_0$.

Then we enter the process when the recovery rate $\beta>0.5$. Comparing (8) with (13), we can find that a small $\gamma_0$ needs less full recovery overhead. If $\gamma_0>0.5$, degree-1 coded symbols are still generated after $\beta>0.5$, but the probability that a degree-1 coded symbol belongs to Case-1 becomes small. Thus the full recovery overhead is larger for a larger $\gamma_0$ when $\gamma_0>0.5$. From the discussion above we can find that setting $\gamma_0>0.5$ results in larger overhead with the same intermediate performance compared with $\gamma_0=0.5$, thus is not recommended.

To sum up, we can formulate the following corollaries.
\begin{corollary}
If $\gamma_0\leq 0.5$, a larger $\gamma_0$ means better intermediate performance but higher full recovery overhead, and vice versa.
\end{corollary}
\begin{corollary}
The value of $\gamma_0$ should not be larger than 0.5.
\end{corollary}
\subsection{Performance Analysis for General Situations}
For any $0<\gamma_0<0.5$ that does not belong to the two special cases, when $0<\beta<\gamma_0$, degree-1 coded symbols are generated, which can be analyzed by Theorem 4. When $\gamma_0<\beta<0.5$, degree-2 coded symbols are generated. It is considered as the first stage of the proposed completion phase. At last, when $0.5\leq \beta \leq 1$, degree-$\hat m$ coded symbols are generated. It is considered as the second stage of the proposed completion phase. First we need to explore how many degree-2 coded symbols become Build-Up edges, which are helpful for the second stage of the completion phase.
\begin{Theorem}
For any $0<\gamma_0<0.5$, denote the number of Build-Up edges introduced by degree-2 coded symbols as $N_B$. $N_B$ can be expressed as
\begin{equation}
N_B=\ln(2-2\gamma_0)kP_u-(\frac12-\gamma_0)k,
\end{equation}
where
\begin{equation}
P_u=\frac12(P_M(\gamma_0k)+P_M(\frac12k)).
\end{equation}
\end{Theorem}
\begin{proof}
For recovering $\frac12k$ source symbols with a given $0<\gamma_0<0.5$, from Theorem 4 we know $-k\ln(1-\gamma_0)$ degree-1 coded symbols are needed, and from Theorem 6 we have $E(N_{\frac12k})=k\ln(2)$. Thus the number of degree-2 coded symbols generated for recovering $\frac12k$ source symbols is
\[
k\ln(2)+k\ln(1-\gamma_0)=\ln(2-2\gamma_0)k.
\]
The exact probability $P_u$ that a generated degree-2 coded symbol is a \emph{useful} symbol is difficult to calculate because it changes with the recovery rate. Hence we use the averaged value of the maximum value of $P_u$ when $\beta=\gamma_0$ and the minimum value of $P_u$ when $\beta=0.5$ as an approximation. Therefore, when $\ln(2-2\gamma_0)k$ degree-2 coded symbols are generated, $\ln(2-2\gamma_0)kP_u$ coded symbols belong to Case-1 or Case-2 and are not discarded, which will contribute to the recovery of the rest $(\frac12-\gamma_0)k$ source symbols in the first stage, or contribute to the recovery in the second stage as Build-Up edges. It is clear that $(\frac12-\gamma_0)k$ Case-1 or Case-2 completion symbols are used for recovering $(\frac12-\gamma_0)k$ source symbols, and the rest are Build-Up edges.
\end{proof}

Then we can analyze the second stage of the completion phase.
\begin{Theorem}
For any $0<\gamma_0<0.5$, denote the number of required coded symbols for recovering $s_3$ source symbols at the second stage of the completion phase as $N_{general3}(s_3)$. The expectation of $N_{general3}(s_3)$ can be expressed as follows:
\begin{equation}
E(N_{general3}(s_3))=(1-\frac{2N_B}{k})\cdot\sum_{i=\frac12k}^{s_3+\frac12k-1}{\frac{1}{P_M(i)}}.
\end{equation}
\end{Theorem}
\begin{proof}
From Lemma 1, in order to recover one source symbol at the second stage of the completion phase, one \emph{useful} symbol is required on average. On the other hand, there are $N_B$ Build-Up edges that can be used to recover $\frac12k$ source symbols. Thus there are $\frac{2N_B}{k}$ Build-Up edges for each source symbol on average. The number of Case-1 and Case-2 completion symbols required for recovering one source symbol is $1-\frac{2N_B}{k}$. Similar to the proof of Theorem 3 and Theorem 5, by using induction we can obtain Theorem 8.
\end{proof}

Lastly we tackle the first stage of the proposed completion phase. We already know that in this stage $\ln(2-2\gamma_0)k$ degree-2 coded symbols are generated, and $(\frac12-\gamma_0)k$ source symbols are recovered. Therefore, on average,
\begin{equation}
\frac{(\frac12-\gamma_0)k}{\ln(2-2\gamma_0)k}=\frac{\frac12-\gamma_0}{\ln(2-2\gamma_0)}
\end{equation}
degree-2 coded symbols are required to recover one source symbol. We use this as an approximation, and obtain the following corollary.
\begin{corollary}
For any $0<\gamma_0<0.5$, denote the number of required coded symbols for recovering $s_2$ source symbols at the first stage of the completion phase as $N_{general2}(s_2)$. The expectation of $N_{general2}(s_2)$ can be expressed as
\begin{equation}
E(N_{general2}(s_2))=\frac{s_2\cdot \ln(2-2\gamma_0)}{\frac12-\gamma_0}.
\end{equation}
\end{corollary}

Combing Theorem 4, Theorem 8 and Corollary 4, we can obtain analysis for the whole encoding process.
\begin{Remark}
For any $0<\gamma_0<0.5$, denote the number of coded symbols required for recovering $s$ source symbols as $N_{g}(s)$, the expectation of $N_{g}(s)$ is expressed as follows:
\begin{align}
E(N_{g}(s))=
&\left\{
\begin{aligned}
&k\ln(\frac{k}{k-s}),
\text{when } 0<s\leq \gamma_0k;\\
&\frac{(s-\gamma_0k)\cdot \ln(2-2\gamma_0)}{\frac12-\gamma_0}-k\ln(1-\gamma_0),
\text{when } \gamma_0k<s\leq \frac12k;\\
&k\ln(2)+(1-\frac{2N_B}{k})\cdot\sum_{i=\frac12k}^{s-1}{\frac{1}{P_M(i)}},
\text{when } \frac12k<s \leq k.
\end{aligned}
\right.
\end{align}
\end{Remark}

Note that for simplicity, the theoretical analysis is derived under lossless channels. However, it is straightforward to extend the results to lossy channels. Because the source symbols are selected uniformly and randomly, the source symbol degree distribution is Poisson distribution, and it remains Poisson distribution after random loss. Thus denote the channel erasure rate $\epsilon$, the aforementioned analysis results just need to be extra divided by $1-\epsilon$ under lossy channels.
\section{Systematic online fountain codes}
In this section, we propose systematic online fountain codes (SOFC) motivated by analysis in Section III. We then analyze the performance of the proposed SOFC, and compare it with the performance of OFC in \cite{cassuto2015online}.

\subsection{Motivation and the Proposed Coding Scheme}
OFCNB can trade-off performance between intermediate performance and full recovery overhead.
However, affected by the \emph{random selection} rule of source symbols, the performance of OFCNB is still not optimal, and can be improved. When $\gamma_0$ is small, as we mentioned in Theorem 2, the random selection results in phase transition \cite{alon2004probabilistic}, before which a large number of small components are built and they will not turn black immediately. Thus the improvement on intermediate performance is marginal. When $\gamma_0$ is large, as we mentioned in Theorem 4, since randomly selected degree-1 coded symbols are more likely to correspond to duplicate source symbols, the full recovery overhead is larger than OFC.


Motivated by the observation from OFCNB, we further investigate the impact of random selection on the intermediate performance and the full recovery overhead. Note that good intermediate performance indicates that many source symbols are quickly recovered. Then the coded symbols generated afterwards are more likely to be the XOR of several source symbols that have been already recovered, i.e., not \emph{useful}, since the source symbols are selected randomly. This results in large full recovery overheads. For the same reason, if we generate more \emph{useful} coded symbols that are randomly selected in order to reduce full recovery overhead, the recovery rate remains to be low for a long time. This degrades the intermediate performance. Therefore, we consider the random selection of source symbols as the root of the trade-off between the intermediate performance and the full recovery overhead, and it prevents us from improving both these two performance at the same time.

Thus we attempt to modify the selection rule of source symbols to reduce the effect of random selection. For the firstly transmitted degree-1 coded symbols, we propose to select source symbols only once. Hence there will be no redundant coded symbols. Since the degree-1 coded symbols are always \emph{useful} now, they shall be generated as many as possible. Therefore we propose the systematic coding scheme as follows.
\\\textbf{Systematic phase}: The encoder selects and transmits source symbols one-by-one, with successive indexes, as a degree-1 coded symbol, until all source symbols are selected once.
\\\textbf{Completion phase}: The decoder calculates the recovery rate $\beta$ and feeds it back to the encoder. The encoder calculates the optimal coded symbols degree $\hat m$ based on (\ref{degree}), and randomly selects $\hat m$ source symbols to generate a coded symbol. The rest of the completion phase is the same as Section II-B, and continues until all source symbols are recovered.

\subsection{Performance Analysis and Comparison}
Now we analyze the performance of SOFC. Since the selection rule is not random, the performance of SOFC is affected by the channel erasure rate $\epsilon$. We start with the analysis in lossy channels.
\begin{Theorem}
For SOFC, in the systematic phase, denote the number of coded symbols required for recovering $s_1$ source symbols as $N_{sys1}(s_1)$. The expectation of $N_{sys1}(s_1)$ can be expressed as
\begin{equation}
E(N_{sys1}(s_1))=\frac{s_1}{1-\epsilon}.
\end{equation}
\end{Theorem}
\begin{proof}
Note that in the systematic phase, once a degree-1 coded symbol is successfully received, it can recover one source symbol.
\end{proof}

$\epsilon$ has a greater impact on the completion phase. First we discuss the situation where $0\leq\epsilon \leq 0.5$, since the other two situations, i.e., $0.5<\epsilon <1$ and $\epsilon\to 1$, are similar to OFCNB. If $0\leq\epsilon \leq 0.5$, we have $1-\epsilon\geq0.5$, thus degree-2 coded symbols will not be generated in the completion phase.
\begin{Theorem}
For SOFC, in the completion phase, denote the number of required coded symbols for recovering $s_2$ source symbols as $N_{sys2}(s_2)$. When $0\leq\epsilon \leq 0.5$, the expectation of $N_{sys2}(s_2)$ can be expressed as
\begin{equation}
E(N_{sys2}(s_2))=\frac{1}{1-\epsilon}\cdot\sum^{s+(1-\epsilon)k-1}_{i=(1-\epsilon)k}\frac{1}{P_M(i)},
\end{equation}
where $P_M(i)$ is obtained in (9).
\end{Theorem}
\begin{proof}
When the systematic phase ends, $k$ degree-1 coded symbols are generated and transmitted, $(1-\epsilon)k$ coded symbols are received, and $(1-\epsilon)k$ source symbols are recovered. Because all coded symbols in this phase belong to Case-1, there are no Build-Up edges introduced. Thus from Lemma 1, in the completion phase, recovering one source symbol requires one Case-1 or Case-2 completion symbol on average. When $(1-\epsilon)k$ source symbols are recovered, the probability that a coded symbol is a Case-1 or Case-2 completion symbol can be expressed as $P_M((1-\epsilon)k)$. By using induction we can obtain Theorem 10.
\end{proof}

When $0.5<\epsilon<1$, the performance is similar to OFCNB with $0<\gamma_0<0.5$, the completion phase can be divided into two stages, and degree-2 coded symbols are generated in the first stage and the Build-Up edges will be introduced. When $\epsilon \to 1$,  the first stage of the completion phase is similar to OFCNB with $\gamma_0\to 0$. Denote the number of coded symbols required for recovering $s_2$ source symbols at the first stage of the completion phase as $N_{sys2}(s_2)$, and the number of coded symbols required for recovering $s_3$ source symbols at the second stage of the completion phase as $N_{sys3}(s_3)$, we derive the following two corollaries.
\begin{corollary}
When $0.5<\epsilon<1$, the expectation of $N_{sys2}(s_2)$ can be expressed as
\begin{equation}
E(N_{sys2}(s_2))=\frac{s_2\cdot\ln(2\epsilon)}{(\epsilon-\frac12)(1-\epsilon)},
\end{equation}
and the expectation of $N_{sys3}(s_3)$ can be expressed as
\begin{equation}
E(N_{sys3}(s_3))=\frac{k-2N_{B\epsilon}}{k(1-\epsilon)}\cdot\sum_{i=\frac12k}^{s_2+(1-\epsilon)k-1}{\frac{1}{P_M(i)}},
\end{equation}
where
\begin{equation}
N_{B\epsilon}=\ln(2\epsilon)kP_{u\epsilon}-(\epsilon-\frac12)k,
\end{equation}
and
\begin{equation}
P_{u\epsilon}=\frac{P_M((1-\epsilon)k)+P_M(\frac12k)}{2}.
\end{equation}
\end{corollary}
\begin{corollary}
When $\epsilon\to 1$, the expectation of $N_{sys2}(s_2)$ can be expressed as
\begin{equation}
E(N_{sys2}(s_2))=-\frac{k^2\ln(1-\frac{s_2}{k})}{2s_2(1-\epsilon)},
\end{equation}
and the expectation of $N_{sys3}(s_3)$ can be expressed as
\begin{equation}
E(N_{sys3}(s_3))=\frac{1-\frac 14c_0}{1-\epsilon}\cdot \sum_{i=\frac12k}^{s_3+\frac12k-1}{\frac{1}{P_M(i)}}.
\end{equation}
\end{corollary}

\begin{proof}
Note that in SOFC, after the systematic phase, $(1-\epsilon)k$ source symbols are recovered. While for OFCNB with $0<\gamma_0<0.5$, $\gamma_0k$ source symbols are recovered. Substitute $\gamma_0$ by $1-\epsilon$ we can obtain (24) to (27). The completion phase of SOFC is the same as OFCNB with $\gamma_0\to 0$ when $\epsilon \to 1$.
\end{proof}

Combing Theorem 9, Theorem 10, Corollary 5 and Corollary 6 we have the following remark.
\begin{Remark}
For SOFC, given the channel erasure rate $\epsilon$, denote the number of required coded symbols for recovering $s$ source symbols as $N_{sys}(s)$. When $0<\epsilon\leq 0.5$, the expectation of $N_{sys}(s)$ is
\begin{align}
E(N_{sys}(s))=
&\left\{
\begin{aligned}
&\frac{s}{1-\epsilon},&\text{when } 0<s\leq (1-\epsilon)k;\\
&k+\frac{1}{1-\epsilon}\cdot\sum^{s-1}_{i=(1-\epsilon)k}\frac{1}{P_M(i)},&\text{when } (1-\epsilon)k<s \leq k.
\end{aligned}
\right.
\end{align}
When $0.5<\epsilon<1$, we have
\begin{align}
E(N_{sys}(s))=
&\left\{
\begin{aligned}
&\frac{s}{1-\epsilon},\text{when } 0<s\leq (1-\epsilon)k;&\\
&k+\frac{(s-(1-\epsilon) k)\cdot\ln(2\epsilon)}{(\epsilon-\frac12)(1-\epsilon)},\text{when } (1-\epsilon)k<s \leq \frac12k;&\\
&k+\frac{k\ln(2\epsilon)}{1-\epsilon}+\frac{k-2N_{B\epsilon}}{k(1-\epsilon)}\cdot\sum^{s-1}_{i=(1-\epsilon)k}\frac{1}{P_M(i)},\text{when } \frac12k<s \leq k.&
\end{aligned}
\right.
\end{align}
When $\epsilon\to 1$ we have
\begin{align}
E(N_{s}(s))=
&\left\{
\begin{aligned}
&s,\text{when } 0<s \leq \gamma_0k;\\
&k-\frac{k^2\ln(1-\frac{s}{k})}{2s(1-\epsilon)},\text{when } \gamma_0k<s\leq \frac12k;\\
&k+\frac{k\ln(2)}{1-\epsilon}+\frac{1-\frac 14c_0}{1-\epsilon} \cdot \sum_{i=\frac12k}^{s-1}{\frac{1}{P_M(i)}},\text{when } \frac12k<s \leq k.
\end{aligned}
\right.
\end{align}
\end{Remark}

Now we can compare the performance of SOFC, OFCNB and OFC. Because a large number of non-duplicated degree-1 coded symbols are generated, clearly SOFC has better intermediate performance than OFC and OFCNB even with large $\gamma_0$. However, the full recovery overhead of SOFC is affected by the channel erasure rate $\epsilon$. The following theorem compares the full recovery overhead.
\begin{Theorem}
When $\epsilon>\epsilon_0$, we have $E(N_{sys}(k))\leq E(N_{OFC}(k))$, where
\begin{equation}
\epsilon_0=\frac12-\frac18c_0.
\end{equation}
\end{Theorem}
\begin{proof}
When $\epsilon\to 1$, compare (32) with (11) we can find that the full recovery overhead is the same since $k$ is negligible compared with $\frac{k\ln(2)}{1-\epsilon}+\frac{(1-\frac 14c_0)}{1-\epsilon} \cdot \sum_{i=\frac12k}^{s-1}{\frac{1}{P_M(i)}}$. When $0\leq\epsilon<1$, note that recovering $k$ source symbols requires $k$ \emph{useful} coded symbols, and the probability for generating \emph{useful} coded symbols in the completion phase is the same for both OFC and SOFC. Thus we evaluate the number of \emph{useful} coded symbols in the build-up phase for OFC and in the systematic phase for SOFC. As we mentioned before, the performance of OFC is the same as OFCNB with $\gamma_0\to 0$. For OFC, $k\ln(2)$ coded symbols are required in the build-up phase, $\frac12k$ of them are utilized for recovering the first $\frac12k$ source symbols, and $\frac14c_0\cdot\frac12k=\frac18kc_0$ of them are Build-Up edges that are helpful in the completion phase, and the rest symbols are discarded. Thus in total, there are $\frac12k+\frac18kc_0$ coded symbols \emph{useful}. For SOFC, only systematic coded symbols are ensured to be \emph{useful}. Thus only when more than $\frac12k+\frac18kc_0$ systematic coded symbols are received, we have $E(N_{sys}(k))>E(N_{OFC}(k))$. Therefore we have
\begin{align*}
(1-\epsilon_0)k&=\frac12k+\frac18kc_0\\
\epsilon_0&=\frac12-\frac18c_0.
\end{align*}
Thus Theorem 11 is obtained.
\end{proof}

When $0<\epsilon<\epsilon_0$, SOFC has better intermediate performance and lower full recovery overhead than OFC. When $\epsilon=\epsilon_0$, compared with OFC, SOFC has better intermediate performance and the same full recovery overhead. When $\epsilon_0<\epsilon<1$, SOFC has better intermediate performance but larger full recovery overhead than OFC. When $\epsilon\to 1$, SOFC has the same intermediate performance and full recovery overhead as OFC.
\section{Numerical results and discussion}
In this section, we present the performance of OFCNB and SOFC. We demonstrate that our proposed analyses match well with simulation results. Moreover, we show that OFCNB with $\gamma_0\to 0$ has better intermediate performance with the same full recovery overhead as OFC. We also show that SOFC has better intermediate performance and a lower full recovery overhead than OFC when the channel erasure rate is smaller than $\epsilon_0$. At last, we propose a method to further reduce the number of feedback transmissions without degrading performance.

\subsection{Validation of Proposed Analysis}
\begin{figure}[t]
\centering
\includegraphics[scale=0.45]{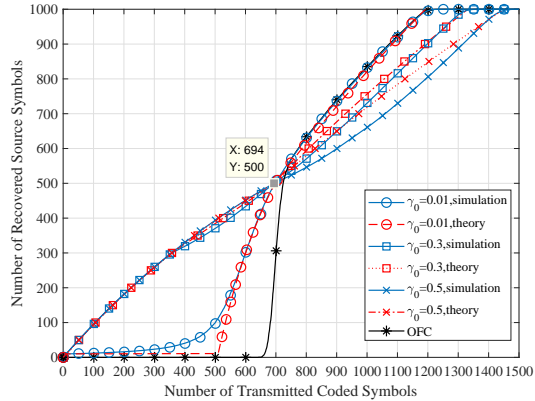}
\caption{The number of transmitted coded symbols versus the number of recovered source symbols for OFCNB.}
\label{OFCNB}
\end{figure}

First we show the theoretical results from Section III. Fig. \ref{OFCNB} shows the relationship between the number of transmitted coded symbols and the number of recovered source symbols for OFCNB when $k=1000$ and $\epsilon=0$. Theoretical and simulation results are presented for OFCNB with $\gamma_0=0.01$, $\gamma_0=0.3$ and $\gamma_0=0.5$, corresponding to three situations that $\gamma_0\to 0$, $0<\gamma_0<0.5$ and $0.5\leq\gamma_0\leq1$, respectively. Simulation results for OFC with $\beta_0=0.5$ is also presented. From Fig. \ref{OFCNB}, first we can see our analysis matches well with simulation results in all of the three situations. We can also find that OFCNB always requires $k\ln(2)\approx 694$ coded symbols for recovering $\frac12 k=500$ source symbols, as we demonstrated in Theorem 6. Compared with OFC, OFCNB with $\gamma_0=0.01$ has better intermediate performance and the same full recovery overheads, as we demonstrated in Corollary 1. Interestingly, when $\beta>0.5$, the curve of OFCNB with $\gamma_0=0.01$ coincides with the curve of OFC, which validates that the second phase of OFCNB completion phase is the same as the completion phase of OFC, as we mentioned in Theorem 3.
With larger $\gamma_0$, the intermediate performance of OFCNB becomes better but the full recovery overhead becomes larger, as we demonstrated in Corollary 2.

\begin{figure}[t]
\centering
\includegraphics[scale=0.45]{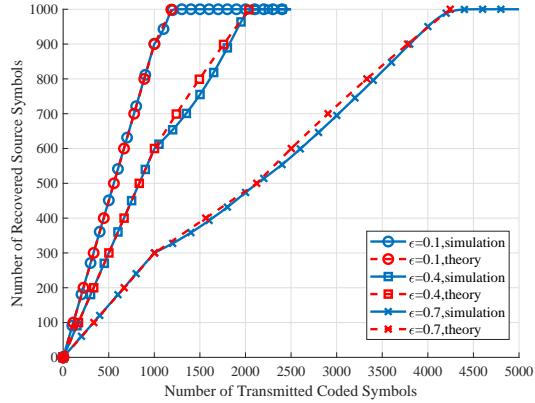}
\caption{The number of transmitted coded symbols versus the number of recovered source symbols for SOFC.}
\label{SOFC}
\end{figure}
\begin{figure}[t]
\centering
\includegraphics[scale=0.45]{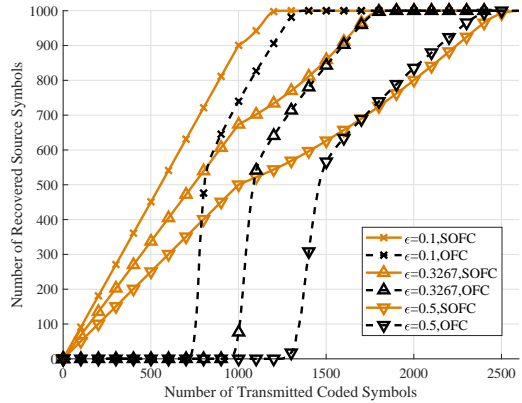}
\caption{The comparison of SOFC and OFC under different channel erasure rate.}
\label{SOFC_OFC}
\end{figure}
Then we present simulation results of SOFC matching with the analysis in Section IV. Fig. \ref{SOFC} shows the relationship between the number of transmitted coded symbols and the number of recovered source symbols of SOFC for $k=1000$. The channel erasure rate is set as $\epsilon=0.1$ and $0.4$ that correspond to $0\leq \epsilon \leq0.5$. We also set channel erasure rate as $\epsilon=0.7$ to correspond to $0.5< \epsilon<1$. The situation of $\epsilon\to 1$ requires infinite overhead for full recovery. Thus it is not demonstrated here. In Fig. \ref{SOFC}, the analysis for SOFC matches well with simulation results. It can be shown that the performance degrades with increasing $\epsilon$. In Fig. \ref{SOFC_OFC}, we compare the performance of SOFC and OFC under channel erasure rate $\epsilon=0.1,0.3267,0.5$. Note that $\epsilon_0=\frac12-\frac18 c_0\approx 0.3267$. As indicated in Theorem 11, when $\epsilon=0.1<\epsilon_0$, the intermediate performance of SOFC is substantially better than OFC, and SOFC also requires less overhead for full recovery. When $\epsilon=\epsilon_0=0.3267$, SOFC still has better intermediate performance than OFC. Yet after the systematic stage, the rate of recovery is not as fast as $\epsilon=0.1$, and the full recovery overhead of SOFC is the same as OFC. When $\epsilon=0.5>\epsilon_0$, the completion phase of SOFC becomes long such that it requires slightly larger overhead for full recovery compared with OFC, while the intermediate performance of SOFC is still much better.

\subsection{Performance Comparisons}
Fig. \ref{BER} presents the BER performance of OFCNB with $\gamma_0=0.01$ and SOFC under different channel erasure rate. The BER performance of OFC is also presented for comparison. We set $k=1000$ and $\epsilon=0.1,0.3$ and $0.5$, respectively. The overhead is defined as $N_t/k$, where $N_t$ is the number of transmitted coded symbols. From Fig. \ref{BER}, we can first see that in low BER regions, OFCNB has the same performance as OFC, while in high BER regions OFCNB has better BER performance than OFC, as has been demonstrated in Fig. \ref{SOFC}. Regarding the performance when $\epsilon=0.1$, SOFC has better BER performance than OFCNB and OFC. When $\epsilon=0.3$, the BER performance of SOFC is notablely better than OFCNB and OFC in high BER regions, and almost the same as OFCNB and OFC in low BER regions. When $\epsilon=0.5$, SOFC has worse BER performance than OFC and OFCNB. Moreover, error floor does not occur for both OFCNB and SOFC when BER is $10^{-6}$.

\begin{figure}[t]
\centering
\includegraphics[scale=0.5]{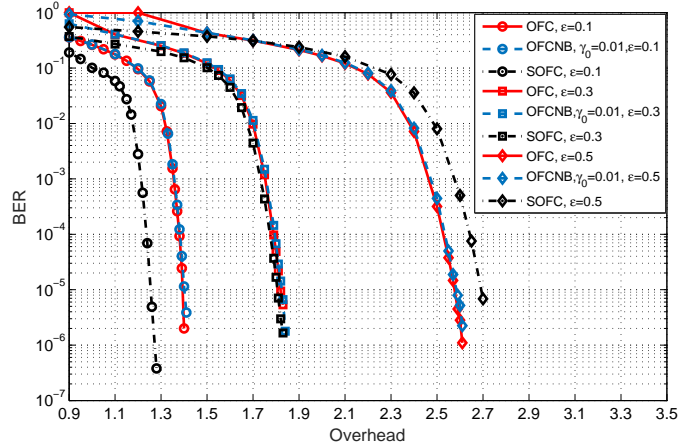}
\caption{BER performance of OFCNB, SOFC and OFC.}
\label{BER}
\end{figure}

In Fig. \ref{512} and Table \ref{feedback} we compare the performance of our proposed schemes with other online fountain coding schemes. We set $k=512$ and $\epsilon=0.1$. For the proposed schemes, the performance of OFCNB with $\gamma_0=0.01$ and SOFC is presented. For other online fountain coding schemes, we present OFC \cite{cassuto2015online}, IOFC \cite{Me2018IOFC}, the scheme in \cite{hashemi2016fountain} with quantized distance $s=50$, and the scheme in \cite{huang2017improved}. We also present the upper bound in Fig. \ref{512} corresponding to the situation where every coded symbol can recover one source symbol once it is received. Among these online fountain coding schemes, IOFC and the scheme in \cite{huang2017improved} improve the full recovery overhead with degraded intermediate performance compared with OFC. The scheme in \cite{hashemi2016fountain} with quantized distance $s=50$ improves the intermediate performance but requires larger overhead than OFC for full recovery. The proposed OFCNB with $\gamma_0=0.01$ also improves the intermediate performance, but it does not require larger overheads than OFC. The proposed SOFC has improvement on both the intermediate performance and the full recovery overhead compare to OFC. We can find that when $\epsilon_0=0.1$, SOFC has better intermediate performance than all other schemes, and for full recovery it requires the same overhead as the best scheme, i.e., the scheme in \cite{huang2017improved}. Moreover, SOFC almost achieves the upper bound when $\epsilon=0.1$.

\begin{figure}[t]
\centering
\includegraphics[scale=0.38]{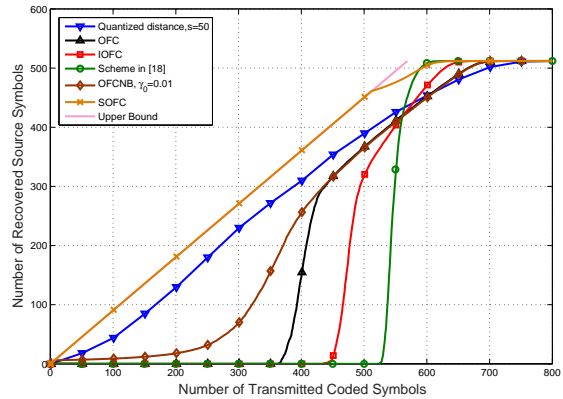}
\caption{Performance comparison of different online fountain coding schemes when $\epsilon=0.1$.}
\label{512}
\end{figure}

\begin{table}[t]
\renewcommand{\arraystretch}{0.6}
\caption{Overhead for full recovery and number of feedback transmissions when $k=512$ and $\epsilon=0.1$}
\label{s502}
\centering
\begin{tabular}{|@{\extracolsep{0.5em}}c @{\extracolsep{0.5em}}|@{\extracolsep{0.1em}}c |@{\extracolsep{0.1em}}c@{\extracolsep{0.1em}}|@{\extracolsep{0.1em}}c@{\extracolsep{0.1em}}|}
\hline
Scheme & Overhead  & Feedback($\beta=0.8$) & Feedback($\beta=1$)\\ \hline
$s=50$ in \cite{hashemi2016fountain}  &1.49 &12.08 &15.1   \\ \hline
OFC  &1.32 &5.33 &21.0   \\ \hline
IOFC  &1.24 &5.02 &16.9   \\ \hline
Scheme in \cite{huang2017improved} &1.18 &3.93 &9.5   \\ \hline
OFCNB($\gamma_0=0.01$) &1.32 &5.99 &29.9   \\ \hline
SOFC &1.18 &0 &22.6   \\ \hline
\end{tabular}
\label{feedback}
\end{table}

Now we compare the number of feedback transmissions required for these schemes, for both $\beta=0.8$ and full recovery, i.e. $\beta=1$. We can find that SOFC requires no feedback when $\beta=0.8$, since the completion phase of SOFC starts when $\beta=0.9$. For full recovery, OFCNB with $\gamma_0=0.01$ requires more feedback transmissions than other schemes, and SOFC requires a number of feedback transmissions comparable to OFC. This is because the optimal degree changes frequently with increasing $\beta$ after $\beta>0.8$, as we can observe from Fig. \ref{pmax}. For OFC, there are several small components in the decoding graph. Thus once the small components turn black, $\beta$ increase with a speed greater than $1/k$. On the other hand, for OFCNB with $\gamma_0=0.01$, there are fewer small components in the decoding graph. Thus the increment of $\beta$ is smaller upon reception of a coded symbol, and resulting in more feedback transmissions than OFC for full recovery. For SOFC, there are no small components in the decoding graph. Thus it also requires more feedback transmissions than OFC for full recovery. However, the completion phase of SOFC is shorter, and it requires fewer feedback transmissions than OFCNB with $\gamma_0=0.01$ for $\beta$ to increase to 1.


\subsection{Further Reduction on Feedback Transmissions}
Note that when $\beta$ is large, although the optimal degree $\hat{m}$ changes frequently with increasing $\beta$, the change of $P_M(\beta)$, i.e., the probability that a coded symbol belongs to Case-1 or Case-2, is negligible and thus the updating of $\hat{m}$ does not contribute significantly to full recovery. Therefore, we propose to update $\hat{m}$ less frequently in order to reduce the number of feedback transmissions. Denote the previously and newly calculated optimal degree as $m_{o}$ and $m_n$, respectively, we propose to set a small threshold $\Delta p$ and only update the coded symbols degree when
\begin{equation}
P_1(m_{n},\beta)+P_2(m_{n},\beta)>P_1(m_{o},\beta)+P_2(m_{o},\beta)+\Delta p.
\end{equation}
We denote the proposed scheme as the threshold-based feedback scheme and evaluate its performance.

In Fig. \ref{th} and Table \ref{thresholdt}, we present the intermediate performance, the full recovery overhead and feedback transmissions required for OFC, OFCNB and SOFC using the threshold-based feedback scheme, with $\Delta p=0.01$, $k=512$ and $\epsilon=0.1$. It is demonstrated in Fig, \ref{th} that the threshold-based feedback scheme has the same intermediate performance and full recovery overhead compared with conventional non-threshold-based scheme. Comparing Table \ref{thresholdt} with Table \ref{feedback}, we can find that threshold-based feedback scheme can significantly reduce the number of feedback transmissions required. In addition, since feedback for SOFC occurs only after $\beta>0.9$ when a newly calculated $\hat{m}$ has negligible impact on $P_M(\beta)$, the reduction of feedback for SOFC is the most significant and it requires less feedback transmissions than OFC with the threshold-based feedback scheme.

\begin{figure}[t]
\centering
\includegraphics[scale=0.45]{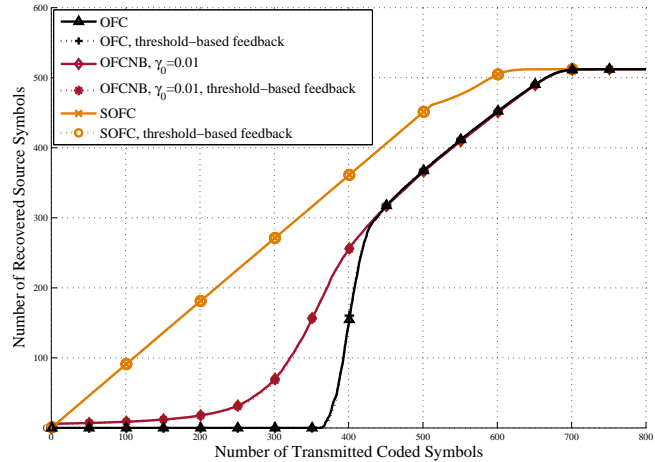}
\caption{Performance of OFC, OFCNB and SOFC with threshold-based feedback scheme when $\Delta p=0.01$ and $\epsilon=0.1$.}
\label{th}
\end{figure}

\begin{table}[t]
\renewcommand{\arraystretch}{0.6}
\caption{Full recovery overhead and number of feedback transmissions for threshold-based feedback scheme when $\Delta p=0.01$, $k=512$ and $\epsilon=0.1$}
\label{s502}
\centering
\begin{tabular}{|@{\extracolsep{0.5em}}c@{\extracolsep{0.5em}}|@{\extracolsep{0.1em}}c@{\extracolsep{0.1em}}|@{\extracolsep{0.1em}}c@{\extracolsep{0.1em}}|}
\hline
Scheme & Overhead  &  Feedback($\beta=1$)\\ \hline
OFC(threshold-based feedback)  &1.32  &14.3   \\ \hline
OFCNB($\gamma_0=0.01$, threshold-based feedback) &1.32 &17.1    \\ \hline
SOFC(threshold-based feedback) &1.18 &9.69   \\ \hline
\end{tabular}
\label{thresholdt}
\end{table}

\section{Conclusions}

We proposed OFCNB to improve the intermediate performance of online fountain codes. By varying the parameter $\gamma_0$, the number of required degree-1 coded symbols can be adjusted, and OFCNB can achieve a trade-off between the intermediate performance and the full recovery overhead. The relationship between the number of transmitted coded symbols and the number of recovered source symbols was analyzed for OFCNB with different $\gamma_0$. We showed that OFCNB with $\gamma_0\to 0$ has better intermediate performance and requires the same full recovery overhead compared with OFC. Motivated by the analysis for OFCNB, we considered the source of the afromentioned trade-off to be the random selection of source symbols, and made OFCNB systematic as SOFC to improve both the intermediate performance and the full recovery overhead. Theoretical analysis for SOFC was then proposed. We showed that SOFC has better intermediate performance for all channel erasure rates, and requires lower full recovery overhead when channel erasure rate is lower than a constant $\epsilon_0$. The proposed analyses were verified by simulation results, and the performance improvement of OFCNB and SOFC was observed in comparison with other online fountain coding schemes.

\balance
\bibliographystyle{IEEEtran}%
\bibliography{../bib/IEEEabrv,../bib/bibfile}

\begin{thebibliography}{10}
\providecommand{\url}[1]{#1}
\csname url@samestyle\endcsname
\providecommand{\newblock}{\relax}
\providecommand{\bibinfo}[2]{#2}
\providecommand{\BIBentrySTDinterwordspacing}{\spaceskip=0pt\relax}
\providecommand{\BIBentryALTinterwordstretchfactor}{4}
\providecommand{\BIBentryALTinterwordspacing}{\spaceskip=\fontdimen2\font plus
\BIBentryALTinterwordstretchfactor\fontdimen3\font minus
  \fontdimen4\font\relax}
\providecommand{\BIBforeignlanguage}[2]{{%
\expandafter\ifx\csname l@#1\endcsname\relax
\typeout{** WARNING: IEEEtran.bst: No hyphenation pattern has been}%
\typeout{** loaded for the language `#1'. Using the pattern for}%
\typeout{** the default language instead.}%
\else
\language=\csname l@#1\endcsname
\fi
#2}}
\providecommand{\BIBdecl}{\relax}
\BIBdecl

\bibitem{Byers1998DF}
J.~W. Byers, M.~Luby, M.~Mitzenmacher, and A.~Rege, ``A digital fountain
  approach to reliable distribution of bulk data,'' \emph{ACM SIGCOMM Computer
  Communication Review}, vol.~28, no.~4, pp. 56--67, October 1998.

\bibitem{Luby2002LT}
M.~Luby, ``\protect{LT} codes,'' in \emph{The 43rd Annual IEEE Symposium on
  Foundations of Computer Science, 2002. Proceedings.}, November, pp. 271--280.

\bibitem{shokrollahi2006raptor}
A.~Shokrollahi, ``Raptor codes,'' \emph{{IEEE} Trans. Inf. Theory}, vol.~52,
  no.~6, pp. 2551--2567, June 2006.

\bibitem{cao2013extended}
C.~Cao, Z.~Fei, M.~Xiao, G.~Huang, C.~Xing, and J.~Kuang, ``An extended
  packetization-aware mapping algorithm for scalable video coding in
  finite-length fountain codes,'' \emph{Science China}, vol.~56, no.~4, pp.
  1--10, 2013.

\bibitem{Hussain2014BDLT}
I.~{Hussain}, M.~{Xiao}, and L.~K. {Rasmussen}, ``Buffer-based distributed
  \protect{LT} codes,'' \emph{{IEEE} Trans. Commun.}, vol.~62, no.~11, pp.
  3725--3739, Nov 2014.

\bibitem{Sorensen2012ripplesize}
J.~H. Sorensen, P.~Popovski, and J.~Ostergaard, ``Design and analysis of
  \protect{LT} codes with decreasing ripple size,'' \emph{{IEEE} Trans.
  Commun.}, vol.~60, no.~11, pp. 3191--3197, November 2012.

\bibitem{Hussain2013MBLT}
I.~Hussain, M.~Xiao, and L.~K. Rasmussen, ``Design of \protect{LT} codes with
  equal and unequal erasure protection over binary erasure channels,''
  \emph{{IEEE} Commun. Lett.}, vol.~17, no.~2, pp. 261--264, February 2013.

\bibitem{Hayajneh2015MBLT}
K.~F. Hayajneh, S.~Yousefi, and M.~Valipour, ``Improved finite-length
  \protect{L}uby-transform codes in the binary erasure channel,'' \emph{IET
  Communications}, vol.~9, no.~8, pp. 1122--1130, May 2015.

\bibitem{Me2017SMBLT}
J.~{Huang}, Z.~{Fei}, D.~{Jia}, C.~{Sun}, and X.~{Wang}, ``Memory based
  \protect{LT} code with shifted degree,'' in \emph{2017 IEEE 17th
  International Conference on Communication Technology (ICCT)}, Oct 2017, pp.
  106--111.

\bibitem{Yuan2008SLT}
X.~{Yuan} and L.~{Ping}, ``On systematic \protect{LT} codes,'' \emph{{IEEE}
  Commun. Lett.}, vol.~12, no.~9, pp. 681--683, Sep. 2008.

\bibitem{Xu2016SLT}
S.~{Xu} and D.~{Xu}, ``Optimization design and asymptotic analysis of
  systematic \protect{L}uby transform codes over \protect{BIAWGN} channels,''
  \emph{{IEEE} Trans. Commun.}, vol.~64, no.~8, pp. 3160--3168, Aug 2016.

\bibitem{Okpotse2019SFC}
T.~{Okpotse} and S.~{Yousefi}, ``Systematic fountain codes for massive storage
  using the truncated poisson distribution,'' \emph{{IEEE} Trans. Commun.},
  vol.~67, no.~2, pp. 943--954, Feb 2019.

\bibitem{kamra2006growth}
A.~Kamra, V.~Misra, J.~Feldman, and D.~Rubenstein, ``Growth codes: Maximizing
  sensor network data persistence,'' in \emph{ACM SIGCOMM Computer
  Communication Review}, vol.~36, no.~4, October 2006, pp. 255--266.

\bibitem{beimel2007rt}
A.~Beimel, S.~Dolev, and N.~Singer, ``\protect{RT} oblivious erasure
  correcting,'' \emph{{IEEE/ACM} Trans. Netw.}, vol.~15, no.~6, pp. 1321--1332,
  December 2007.

\bibitem{cassuto2015online}
Y.~Cassuto and A.~Shokrollahi, ``Online fountain codes with low overhead,''
  \emph{{IEEE} Trans. Inf. Theory}, vol.~61, no.~6, pp. 3137--3149, 2015.

\bibitem{Me2017UEPonline}
J.~Huang, Z.~Fei, C.~Cao, M.~Xiao, and D.~Jia, ``On-line fountain codes with
  unequal error protection,'' \emph{{IEEE} Commun. Lett.}, vol.~21, no.~6, pp.
  1225--1228, June 2017.

\bibitem{Cai2019URT}
P.~{Cai}, Y.~{Zhang}, C.~{Pan}, and J.~{Song}, ``Online fountain codes with
  unequal recovery time,'' \emph{{IEEE} Commun. Lett.}, vol.~23, no.~7, pp.
  1136--1140, July 2019.

\bibitem{huang2017improved}
T.~Huang and B.~Yi, ``Improved online fountain codes based on shaping for left
  degree distribution,'' \emph{AEU-International Journal of Electronics and
  Communications}, 2017.

\bibitem{Zhao2018IETonline}
Y.~{Zhao}, Y.~{Zhang}, F.~C.~M. {Lau}, H.~{Yu}, and Z.~{Zhu}, ``Improved online
  fountain codes,'' \emph{IET Communications}, vol.~12, no.~18, pp. 2297--2304,
  2018.

\bibitem{Me2018IOFC}
J.~{Huang}, Z.~{Fei}, C.~{Cao}, M.~{Xiao}, and D.~{Jia}, ``Performance analysis
  and improvement of online fountain codes,'' \emph{{IEEE} Trans. Commun.},
  vol.~66, no.~12, pp. 5916--5926, Dec 2018.

\bibitem{Yi2018onlineData}
B.~Yi, M.~Xiang, T.~Huang, H.~Huang, K.~Qiu, and W.~Li, ``Data gathering with
  distributed rateless coding based on enhanced online fountain codes over
  wireless sensor networks,'' \emph{AEU-International Journal of Electronics
  and Communications}, vol.~92, pp. 86--92, 2018.

\bibitem{hashemi2016fountain}
M.~Hashemi, Y.~Cassuto, and A.~Trachtenberg, ``Fountain codes with nonuniform
  selection distributions through feedback,'' \emph{{IEEE} Trans. Inf. Theory},
  vol.~62, no.~7, pp. 4054--4070, July 2016.

\bibitem{alon2004probabilistic}
N.~Alon and J.~H. Spencer, \emph{The probabilistic method}.\hskip 1em plus
  0.5em minus 0.4em\relax John Wiley \& Sons, 2004.

\end{thebibliography}

\end{document}